\documentclass[11pt]{article}

\usepackage{amsthm}
\usepackage{amsmath}
\usepackage[english]{babel}
\usepackage{amssymb}
\usepackage{mathtools}
\usepackage[pdfencoding=auto]{hyperref}

\usepackage[margin=1in]{geometry}
\usepackage{subcaption}
\newcommand{\dist}{\operatorname{dist}}

\newtheorem{theorem}{Theorem}

\newtheorem{lemma}{Lemma}

\theoremstyle{definition}
\newtheorem{example}{Example}
\newtheorem{definition}{Definition}
\theoremstyle{definition}

\theoremstyle{definition}

\DeclareMathOperator{\Sym}{Sym}

\newcommand{\restrictedto}[1]{|_{#1}}

\usepackage{tikz}

\def\bE{\mathbb{E}}
\def\bP{\mathbb{P}}

\def\bZ{\mathbb{Z}}

\def\cD{\mathcal{D}}
\def\cM{\mathcal{M}}

\def\cS{\mathcal{S}}
\def\cV{\mathcal{V}}
\def\cX{\mathcal{X}}

\def\wt{\widetilde}

\DeclareMathOperator\TV{\mathrm{TV}}

\providecommand{\keywords}[1]
{
{\small \textbf{\textit{Keywords---}} #1}
}

\begin{document}

\title{Generalized Rainbow Differential Privacy}

\date{}

\author{
Yuzhou Gu\thanks{\texttt{yuzhougu@ias.edu}. School of Mathematics, Institute for Advanced Study, Princeton, NJ 08540}
\and
Ziqi Zhou\thanks{\texttt{zhou@ccs-labs.org}. Telecommunication Networks Group (TKN), Department of Telecommunication Systems, Technische Universit{\"a}t Berlin, 10587 Berlin, Germany}
\and
Onur G{\"u}nl{\"u}\thanks{\texttt{onur.gunlu@liu.se}. Information Coding Division, Department of Electrical Engineering, Link{\"o}ping University, 58183 Link{\"o}ping, Sweden}
\and
Rafael G. L. D'Oliveira\thanks{\texttt{rdolive@clemson.edu}. School of Mathematical and Statistical Sciences, Clemson University, Clemson, SC 29634}
\and
Parastoo Sadeghi\thanks{\texttt{p.sadeghi@unsw.edu.au}. School of Engineering and Technology, The University of New South Wales, Canberra, Australia}
\and
Muriel M\'edard\thanks{\texttt{medard@mit.edu}. Research Laboratory of Electronics (RLE), Massachusetts Institute of Technology, Cambridge, MA 02139}
\and
Rafael F. Schaefer\thanks{\texttt{rafael.schaefer@tu-dresden.de}. Chair of Information Theory and Machine Learning, the BMBF Research Hub 6G-life, the Cluster of Excellence “Centre for Tactile Internet with Human-in-the-Loop (CeTI)”, and the 5G Lab Germany, Technische Universit{\"a}t Dresden, 01062 Dresden, Germany}
}

\maketitle

\begin{abstract}
We study a new framework for designing differentially private (DP) mechanisms via randomized graph colorings, called rainbow differential privacy. In this framework, datasets are nodes in a graph, and two neighboring datasets are connected by an edge. Each dataset in the graph has a preferential ordering for the possible outputs of the mechanism, and these orderings are called rainbows. Different rainbows partition the graph of connected datasets into different regions. We show that if a DP mechanism at the boundary of such regions is fixed and it behaves identically for all same-rainbow boundary datasets, then a unique optimal $(\epsilon,\delta)$-DP mechanism exists (as long as the boundary condition is valid) and can be expressed in closed-form. Our proof technique is based on an interesting relationship between dominance ordering and DP, which applies to any finite number of colors and for $(\epsilon,\delta)$-DP, improving upon previous results that only apply to at most three colors and for $\epsilon$-DP. We justify the homogeneous boundary condition assumption by giving an example with non-homogeneous boundary condition, for which there exists no optimal DP mechanism.

\end{abstract}


\section{Introduction}
Differential privacy (DP) is a general framework that aims to limit the statistical capability of a curious analyst, irrespective of its computational power, in determining whether or not the data of a specific participant was used in response to its query\footnote{DP variants assuming a finite computational power for the adversary have been studied in works including \cite{computationalDP} but are not within the scope of our work.} \cite{dwork2016calibrating, dworkoriginal2006}; see \cite{DworkBook} for a treatment of the subject and \cite{survey_2017} for a survey. Recently, DP is applied in the 2020 US Census \cite{uscensus}, as well as by Apple, Google, and Microsoft~\cite{google,apple,microsoft}.

DP imposes constraints on all neighboring datasets, which traditionally differ only in data from one participant. These constraints are \emph{relative} (specified as ratios of mechanism probability distributions), \emph{local} (specified for neighboring datasets), and \emph{dataset-independent} (agnostic to the underlying data distribution or structure)~\cite{DworkBook}, contributing to the success of DP. However, many DP implementations are agnostic to the actual dataset at hand, which is nonadaptive and undesirable~\cite{individual}. Majority of output perturbation DP mechanisms consider the worst-case query sensitivity between any two neighboring datasets to determine the scale of noise~\cite{DworkBook}. This approach is pessimistic and can negatively affect the query utility~\cite{individual}.

Several solutions are available that improve the query utility. For instance, noise calibration was proposed in \cite{nissim2007smooth} to smooth the sensitivity, but a chosen utility level is then not guaranteed and the mechanism suffers from a heavy tail leading to outliers. Another direction is to relax the DP constraints \cite{individual, blowfish, profile}. For example, \cite{individual} proposed \emph{individual}-DP that defines DP constraints only for \emph{given} datasets and their neighbors. The individual-DP framework destroys the group DP, i.e., DP constraints for non-neighboring datasets are no longer valid.

Recently, \cite{RafnoDP2colorPaper} proposed a method to design dataset-dependent DP mechanisms for binary-valued queries that guarantee optimal utility without weakening the original DP constraints; see also \cite{BinaryOutputDP}. In the model in~\cite{RafnoDP2colorPaper}, each dataset has a true query value (e.g., \texttt{blue} or \texttt{red}) and is represented as a node on a graph with edges, representing neighboring datasets. Moreover, they consider DP mechanisms which act homogeneously at the boundary datasets.\footnote{Boundary datasets are neighbors whose true query value is different from each other.} They then show how these initial constraints can be optimally extended in closed-form for all other datasets, where the probability of giving the truthful query response is maximized by taking into account the distance to the boundary.

The framework in \cite{RafnoDP2colorPaper} was generalized in \cite{zhou2022rainbow} by increasing the number of possible query outputs to three (e.g., \texttt{blue}, \texttt{red}, and \texttt{green} that represent majority votes among three choices). This extension is challenging in several ways. In the binary case, the optimal probability assignment for one color (e.g., \texttt{blue}) automatically determines the whole mechanism. In the multi-color case, this is not possible. Thus, a preferential order of colors at each dataset is assumed to solve this problem, in which a mechanism is defined to be better than another if the preferred colors are output with larger probabilities. When there are at most three colors, it is shown that for a DP mechanism that is homogeneous at the boundary, at most one optimal $\epsilon$-DP mechanism exists, for which a closed-form expression is also given. This result recovers the binary case of \cite{RafnoDP2colorPaper} as a special case.

\subsection{Main Contributions}
In this work, we significantly improve \cite{zhou2022rainbow} by providing a new proof technique that allows us to extend the results to any number of colors and any $(\epsilon,\delta)$-DP requirements. We show that given a valid boundary homogeneous DP mechanism, at most one optimal $(\epsilon,\delta)$-DP mechanism exists, for which we provide a closed-form expression. Our results recover the result of \cite{zhou2022rainbow}, in which there are three colors and $\delta=0$. We note that our definition of optimality of a mechanism is through dominance ordering (see Definition~\ref{def:order_reasonable} below), while in \cite{zhou2022rainbow}, optimality is defined through lexicographic ordering. Because dominance ordering is stronger than lexicographic ordering (i.e., $x\preceq y$ in dominance ordering implies $x\preceq y$ in lexicographic ordering, but two elements comparable in lexicographic ordering are not necessarily comparable in dominance ordering), our optimality result is strictly stronger than~\cite{zhou2022rainbow} even in the ternary and $\delta=0$ case.

At its core, our proof uses an interesting relationship between DP and dominance ordering. Namely, for any $\epsilon,\delta\ge 0$ and any distribution $P$ on an ordered set $\cV$, there exists a unique distribution $Q$ that is $(\epsilon,\delta)$-close to $P$ and that dominates any other distribution $Q'$ that is $(\epsilon,\delta)$-close to $P$; see Section~\ref{sec:line-graph} below for more details. Finally, we justify the homogeneous boundary condition assumption by presenting an example with a non-homogeneous boundary condition, such that there exist valid DP mechanisms, but there are no optimal DP mechanisms; see Example~\ref{eg:no-optimal} below.
\subsection{Organization of the Paper}
In Section~\ref{sec:setting}, we introduce the setting for rainbow DP.
In Section~\ref{sec:opt}, we show that to construct optimal DP mechanisms for general graphs under homogeneous boundary conditions, it suffices to do so for a special class of graphs called line graphs. Furthermore, we show that an optimal DP mechanism may not exist for non-homogeneous boundary conditions, and discuss the relationship between our optimality condition and that of previous work \cite{zhou2022rainbow}.
In Section~\ref{sec:line-graph}, we construct optimal DP mechanisms for line graphs.
In Section~\ref{sec:comp-optimal-dp}, we give explicit formulas for the optimal rainbow DP mechanism and present several examples.
In Section~\ref{sec:discussion}, we summarize our results, discuss related approaches and possible further directions.
\subsection{Notation}
All logarithms in this paper are natural logarithms unless otherwise noted.
For a non-negative integer $n$, we use $[n]$ to denote the set $\{1,\ldots,n\}$. For two integers $n\le m$, we use $[n:m]$ to denote the set $\{n,\ldots,m\}$.
For two distributions $P, Q$ on a measurable space $\cX$, we define their total variation (TV) distance as
\begin{align}
\TV(P, Q) = \sup_{\cS} |P(\cS)-Q(\cS)|,
\end{align}
where $\cS$ goes over measurable subsets of $\cX$.

\usetikzlibrary{positioning}
\tikzset{w/.pic={
    \clip (0,0) circle (0.25);
    \begin{scope}
        \fill[white] (-0.25,0) rectangle (0.25,0.25);
        \fill[white] (-0.25,-0.25) rectangle (0.25,0);
    \end{scope}
    \draw (0,0) circle (0.25);
    \path [draw=white,line width=5pt](-0.25,0)--(0.25,0);}}

\usetikzlibrary{positioning}
\tikzset{brg/.pic={
    \clip (0,0) circle (0.25);
    \begin{scope}
        \fill[blue] (-0.25,0) rectangle (0.25,0.25);
        \fill[green] (-0.25,-0.25) rectangle (0.25,0);
    \end{scope}
    \draw (0,0) circle (0.25);
    \path [draw=red,line width=5pt](-0.25,0)--(0.25,0);}}

\tikzset{bgr/.pic={
    \clip (0,0) circle (0.25);
    \begin{scope}
        \fill[blue] (-0.25,0) rectangle (0.25,0.25);
        \fill[red] (-0.25,-0.25) rectangle (0.25,0);
    \end{scope}
    \draw (0,0) circle (0.25);
    \path [draw=green,line width=5pt]
(-0.25,0) -- (0.25,0);}}

\tikzset{rgb/.pic={
    \clip (0,0) circle (0.25);
    \begin{scope}
        \fill[red] (-0.25,0) rectangle (0.25,0.25);
        \fill[blue] (-0.25,-0.25) rectangle (0.25,0);
    \end{scope}
    \draw (0,0) circle (0.25);
    \path [draw=green,line width=5pt]
(-0.25,0) -- (0.25,0);}}

\tikzset{rbg/.pic={
    \clip (0,0) circle (0.25);
    \begin{scope}
        \fill[red] (-0.25,0) rectangle (0.25,0.25);
        \fill[green] (-0.25,-0.25) rectangle (0.25,0);
    \end{scope}
    \draw (0,0) circle (0.25);
    \path [draw=blue,line width=5pt]
(-0.25,0) -- (0.25,0);}}

\tikzset{gbr/.pic={
    \clip (0,0) circle (0.25);
    \begin{scope}
        \fill[green] (-0.25,0) rectangle (0.25,0.25);
        \fill[red] (-0.25,-0.25) rectangle (0.25,0);
    \end{scope}
    \draw (0,0) circle (0.25);
    \path [draw=blue,line width=5pt]
(-0.25,0) -- (0.25,0);}}

\tikzset{grb/.pic={
    \clip (0,0) circle (0.25);
    \begin{scope}
        \fill[green] (-0.25,0) rectangle (0.25,0.25);
        \fill[blue] (-.25,-0.25) rectangle (0.25,0);
    \end{scope}
    \draw (0,0) circle (0.25);
    \path [draw=red,line width=5pt] (-0.25,0) -- (0.25,0);}}

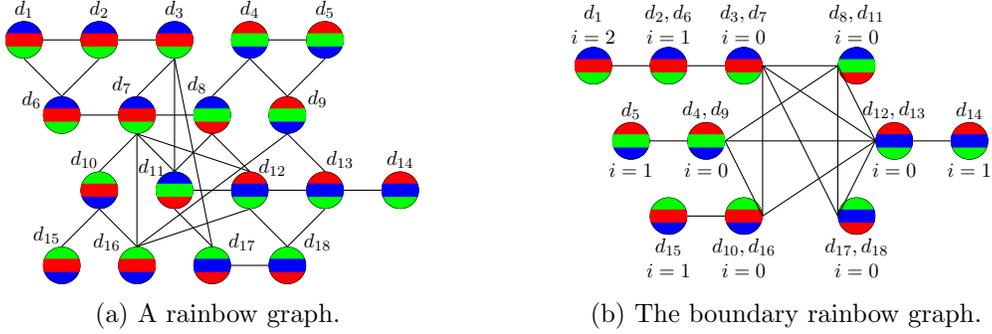
\begin{figure*}[!t]
\centering
\begin{subfigure}[t]{0.45\textwidth}
\centering
\begin{tikzpicture}
    \path pic {brg} node [align=center, above=2mm, scale=0.7, black] {$d_1$};
    \path pic [xshift=1cm] {brg} node [align=center, xshift=1cm, above=2mm, scale=0.7, black] {$d_2$};
    \path pic [xshift=2cm] {brg} node [align=center, xshift=2cm, above=2mm, scale=0.7, black] {$d_3$};
    \path pic [xshift=3cm] {rgb} node [align=center, xshift=3cm, above=2mm, scale=0.7, black] {$d_4$};
    \path pic [xshift=4cm] {rgb} node [align=center, xshift=4cm, above=2mm, scale=0.7, black] {$d_5$};
    \path pic [xshift=0.5cm, yshift=-1cm] {brg} node [align=center, xshift=0.1cm, yshift=-1.2cm, above=2mm, scale=0.7, black] {$d_6$};
    \path pic [xshift=1.5cm, yshift=-1cm] {brg} node [align=center, xshift=1.3cm, yshift=-1.0cm, above=2mm, scale=0.7, black] {$d_7$};
    \path pic [xshift=2.5cm, yshift=-1cm] {bgr} node [align=center, xshift=2.3cm, yshift=-1.0cm, above=2mm, scale=0.7, black] {$d_8$};
    \path pic [xshift=3.5cm, yshift=-1cm] {rgb} node [align=center, xshift=3.9cm, yshift=-1.2cm, above=2.1mm, scale=0.7, black] {$d_9$};
    \path pic [xshift=1cm, yshift=-2cm] {grb} node [align=center, xshift=0.8cm, yshift=-2.0cm, above=2mm, scale=0.7, black] {$d_{10}$};
    \path pic [xshift=2cm, yshift=-2cm] {bgr} node [align=center, xshift=1.7cm, yshift=-2.1cm, above=2mm, scale=0.7, black] {$d_{11}$};
    \path pic [xshift=3cm, yshift=-2cm] {rbg} node [align=center, xshift=3.3cm, yshift=-2.1cm, above=2mm, scale=0.7, black] {$d_{12}$};
    \path pic [xshift=4cm, yshift=-2cm] {rbg} node [align=center, xshift=4.2cm, yshift=-2cm, above=2mm, scale=0.7, black] {$d_{13}$};
    \path pic [xshift=5cm, yshift=-2cm] {rbg} node [align=center, xshift=5cm, yshift=-2cm, above=2mm, scale=0.7, black] {$d_{14}$};
    \path pic [xshift=0.5cm, yshift=-3cm] {grb} node [align=center, xshift=0.3cm, yshift=-3cm, above=2mm, scale=0.7, black] {$d_{15}$};
    \path pic [xshift=1.5cm, yshift=-3cm] {grb} node [align=center, xshift=1.1cm, yshift=-3.1cm, above=2mm, scale=0.7, black] {$d_{16}$};
    \path pic [xshift=2.5cm, yshift=-3cm] {gbr} node [align=center, xshift=2.9cm, yshift=-3.1cm, above=2mm, scale=0.7, black] {$d_{17}$};
    \path pic [xshift=3.5cm, yshift=-3cm] {gbr} node [align=center, xshift=3.9cm, yshift=-3.1cm, above=2mm, scale=0.7, black] {$d_{18}$};
    \draw (0.25,0) -- (0.75,0);
    \draw (1.25,0) -- (1.75,0);
    \draw (3.25,0) -- (3.75,0);
    \draw (0,-0.25) -- (0.5,-0.75);
    \draw (1,-0.25) -- (0.5,-0.75);
    \draw (2,-0.25) -- (1.5,-0.75);
    \draw (2,-0.25) -- (2,-1.75);
    \draw (2,-0.25) -- (2.5,-2.75);
    \draw (3,-0.25) -- (2.5,-0.75);
    \draw (3,-0.25) -- (3.5,-0.75);
    \draw (4,-0.25) -- (3.5,-0.75);
    \draw (0.75,-1) -- (1.25,-1);
    \draw (1.75,-1) -- (2.25,-1);
    \draw (1.5,-1.25) -- (1,-1.75);
    \draw (1.5,-1.25) -- (1.5,-2.75);
    \draw (1.5,-1.25) -- (2,-1.75);
    \draw (1.5,-1.25) -- (3,-1.75);
    \draw (2.5,-1.25) -- (2,-1.75);
    \draw (2.5,-1.25) -- (3,-1.75);
    \draw (3.5,-1.25) -- (1.5,-2.75);
    \draw (3.5,-1.25) -- (4,-1.75);
    \draw (2.25,-2) -- (2.75,-2);
    \draw (3.25,-2) -- (3.75,-2);
    \draw (4.25,-2) -- (4.75,-2);
    \draw (1,-2.25) -- (0.5,-2.75);
    \draw (1,-2.25) -- (1.5,-2.75);
    \draw (3,-2.25) -- (1.5,-2.75);
    \draw (2,-2.25) -- (2.5,-2.75);
    \draw (3,-2.25) -- (3.5,-2.75);
    \draw (4,-2.25) -- (3.5,-2.75);
    \draw (2.75,-3) -- (3.25,-3);
    
\end{tikzpicture}
\caption{A rainbow graph.}
\label{}%
\end{subfigure}
\begin{subfigure}[t]{0.45\textwidth}
\centering
\begin{tikzpicture}
\path pic {brg} node [align=center, above=2mm, scale=0.7, black] {$d_1$\\$i=2$};
\path pic [xshift=1cm] {brg} node [align=center, xshift=1cm, above=2mm, scale=0.7, black] {$d_2,d_6$\\$i=1$};
\path pic [xshift=2cm] {brg} node [align=center, xshift=2cm, above=2mm, scale=0.7, black] {$d_3,d_7$\\$i=0$};
\path pic [xshift=3.5cm] {bgr} node [align=center, xshift=3.5cm, above=2mm, scale=0.7, black] {$d_8,d_{11}$\\$i=0$};
\path pic [xshift=0.5cm, yshift=-1cm] {rgb} node [align=center, xshift=0.5cm, yshift=-1cm, above=2mm, scale=0.7, black] {$d_5$} node [align=center, xshift=0.5cm, yshift=-1cm, below=2mm, scale=0.7, black] {$i=1$};
\path pic [xshift=1.5cm, yshift=-1cm] {rgb} node [align=center, xshift=1.5cm, yshift=-1cm, above=2mm, scale=0.7, black] {$d_4,d_9$} node [align=center, xshift=1.5cm, yshift=-1cm, below=2mm, scale=0.7, black] {$i=0$};
\path pic [xshift=4cm, yshift=-1cm] {rbg} node [align=center, xshift=4cm, yshift=-1cm, above=2mm, scale=0.7, black] {$d_{12},d_{13}$} node [align=center, xshift=4cm, yshift=-1cm, below=2mm, scale=0.7, black] {$i=0$};
\path pic [xshift=5cm, yshift=-1cm] {rbg} node [align=center, xshift=5cm, yshift=-1cm, above=2mm, scale=0.7, black] {$d_{14}$} node [align=center, xshift=5cm, yshift=-1cm, below=2mm, scale=0.7, black] {$i=1$};
\path pic [xshift=1cm, yshift=-2cm] {grb} node [align=center, xshift=1cm, yshift=-2cm, below=2mm, scale=0.7, black] {$d_{15}$\\$i=1$};
\path pic [xshift=2cm, yshift=-2cm] {grb} node [align=center, xshift=2cm, yshift=-2cm, below=2mm, scale=0.7, black] {$d_{10}, d_{16}$\\$i=0$};
\path pic [xshift=3.5cm, yshift=-2cm] {gbr} node [align=center, xshift=3.5cm, yshift=-2cm, below=2mm, scale=0.7, black] {$d_{17}, d_{18}$\\$i=0$};

\draw (0.25,0) -- (0.75,0);
\draw (1.25,0) -- (1.75,0);
\draw (0.75,-1) -- (1.25,-1);
\draw (4.25,-1) -- (4.75,-1);
\draw (1.25,-2) -- (1.75,-2);

\draw  (2.25,0) -- (3.25,0);
\draw  (2.25,0) -- (2.25,-2);
\draw  (2.25,0) -- (3.75,-1);
\draw  (2.25,0) -- (3.25,-2);

\draw  (1.75,-1) -- (3.75,-1);
\draw  (1.75,-1) -- (3.25,0);
\draw  (1.75,-1) -- (2.25,-2);

\draw  (2.25,-2) -- (3.75,-1);

\draw  (3.25,-2) -- (3.25,0);
\draw  (3.25,-2) -- (3.75,-1);

\draw  (3.25,0) -- (3.75,-1);

\end{tikzpicture}
\caption{The boundary rainbow graph.}
\label{fig: rainbow b}
\end{subfigure}
\caption{A rainbow graph and its corresponding boundary graph. A vertex represents a dataset and its neighboring datasets are connected by an edge. The function output space is represented by three colors \texttt{blue}, \texttt{red}, and \texttt{green}. Each dataset has a color preference, represented by the ordering inside the vertex. For instance, vertex $d_1$ prefers \texttt{blue} to \texttt{red} and \texttt{red} to \texttt{green}. We call each such color ordering a rainbow. A DP mechanism is then a probability distribution over colors for every vertex. In (B), we show the boundary rainbow graph of the rainbow graph shown in (A), as described in Definition~\ref{def:boundarymorph} below. In Theorem \ref{thm:reduce-to-line} we show how, for homogeneous boundary conditions (defined in Definition \ref{def:boundaryhomogeneity} below), optimal $(\epsilon,\delta)$-DP mechanisms on (A) can be retrieved from optimal ones on (B). For example, for rainbow $c=(\texttt{red}, \texttt{green},\texttt{blue})$, the vertex $(c,0)$ in the boundary rainbow graph corresponds to datasets $d_4,d_9$ in the original rainbow graph because they are on the boundary of $B^c$ in the original graph. There is an edge between $(c=(\texttt{red}, \texttt{green},\texttt{blue}),0)$ and $(c'=(\texttt{red},\texttt{blue},\texttt{green}),0)$ in the boundary rainbow graph, because there is an edge $(d_9,d_{13})$ in the original rainbow graph, with $i\in B^c$, $m\in B^{c'}$.}
\label{fig: rainbow}
\end{figure*}

\section{Rainbow Differential Privacy} \label{sec:setting}
We denote by $(\mathcal{D}, \sim)$ a family of datasets together with a symmetric neighborhood relationship, where $d,d' \in \mathcal{D}$ are neighbors if $d \sim d'$. We consider a finite output space $\mathcal{V}$. Each dataset $d \in \mathcal{D}$ has an ordered preference for the elements of $\mathcal{V}$, captured by what we call a \emph{rainbow} that represents each preference order.

\begin{definition}
Let $\mathcal{V}$ be a finite output space. A rainbow on $\mathcal{V}$ is a total ordering of $\mathcal{V}$. We denote a rainbow as a permutation vector $c \in \Sym(\mathcal{V})$, where $\Sym(\mathcal{V})$ is the set of all permutations of $\mathcal{V}$.
\end{definition}

The preference of a dataset is captured by the \emph{preference function} $f:\mathcal{D} \rightarrow \Sym(\mathcal{V})$ that assigns a rainbow to each dataset $d \in \mathcal{D}$. Thus, if $f(d) = (\texttt{blue, red, green})$, then it means that the dataset $d \in \mathcal{D}$ prefers \texttt{blue} to \texttt{red} and \texttt{red} to \texttt{green}. Moreover, the goal is to construct a random function $\mathcal{M}:\mathcal{D} ~\to ~\mathcal{V}$ that, for each dataset $d \in \mathcal{D}$, randomly puts out an element of $\mathcal{V}$ such that for a given DP constraint a pre-specified utility function is maximized. As commonly done in the DP literature, we refer to the random function as a \emph{mechanism}. A mechanism is DP if the distribution of its output on neighboring datasets are approximately indistinguishable, as we formalize next.

\begin{definition}[\hspace{1sp}\cite{DworkBook}]\label{def:DPdef}
Let $\epsilon,\delta$ be non-negative real numbers with $\delta \leq 1$.
For two distributions $P$ and $Q$ on $\cV$, we say $P$ and $Q$ are $(\epsilon,\delta)$-close if for any $\cS\subseteq \cV$, we have $P(\cS) \le e^\epsilon Q(\cS) + \delta$ and $Q(\cS) \le e^\epsilon P(\cS) + \delta$. A mechanism $\mathcal{M}: \mathcal{D}\rightarrow \mathcal{V}$ is called $(\epsilon,\delta)$-DP if for any $d \sim d'$, the distributions of $\mathcal{M}(d)$ and $\mathcal{M}(d')$ are $(\epsilon,\delta)$-close.
If $\cM$ is $(\epsilon,0)$-DP, then we also say it is $\epsilon$-DP.
We denote the set of all $(\epsilon,\delta)$-DP mechanisms by $\mathfrak{M}$.
\end{definition}

The performance of a mechanism is measured via a utility function $U:\mathfrak{M} \rightarrow \mathbb{R}$, where $U[\mathcal M] \geq U[\mathcal M']$ means that the mechanism $\mathcal M$ outperforms $\mathcal M'$. In this work, we consider utility functions that agree with the preference function $f:\mathcal{D} \rightarrow \Sym(\mathcal{V})$, i.e., it is preferable for a dataset $d \in \mathcal{D}$ to output a color it prefers according to its rainbow $f(d) \in \Sym(\mathcal{V})$.

\begin{definition} \label{def:order_reasonable}
Let $\preceq$ be the dominance ordering on the probability simplex
\begin{align}
\Delta(\mathcal{V}) = \{x \in [0,1]^{|\mathcal{V}|} : x_1+\cdots+x_{|\mathcal{V}|} =~1\},
\end{align}
i.e., for $x,y\in \Delta(\cV)$, $x\preceq y$ if and only if $x_1+\cdots+x_k \le y_1+\cdots+y_k$ for all $1\le k\le |\cV|$.
For every mechanism $\mathcal{M} \in \mathfrak{M}$ and dataset $d \in \mathcal{D}$, let $\vec{\mathcal{M}}(d) \!\!~\in~\!\! \Delta(\mathcal{V})$ be the vector with coordinates $\vec{\mathcal{M}}(d)_k = \bP[\mathcal{M}(d) = f(d)_k]$. Then, a mechanism $\mathcal{M}\in \mathfrak{M}$ dominates another mechanism $\mathcal{M}'\in \mathfrak{M}$ (denoted by $\cM \succeq \cM'$) if for every dataset $d \in \mathcal{D}$, $\vec{\mathcal{M}}(d) \succeq \vec{\mathcal{M}'}(d)$. Moreover, we say a utility function $U:\mathfrak{M} \rightarrow \mathbb{R}$ is \textit{order reasonable} if whenever a mechanism $\mathcal{M}\in \mathfrak{M}$ dominates another mechanism $\mathcal{M}'\in \mathfrak{M}$, we have $U[\mathcal{M}] \geq U[\mathcal{M}']$.
\end{definition}

The notion of domination in Definition~\ref{def:order_reasonable} induces a partial order on the set $\mathfrak{M}$ of all $(\epsilon,\delta)$-DP mechanisms. When a mechanism $\mathcal M$ dominates $\mathcal M'$, it means that $\mathcal M$ outperforms $\mathcal M'$ for any order reasonable utility. In this setting, we say that a mechanism is \emph{optimal} if no other mechanism dominates it.

An interesting subclass of order reasonable utility functions is the set of functions $U$ of the form $\displaystyle U[\cM] = \bE\left[ \sum_{d\in \cD} u_d(\cM(d))\right]$,
where the expectation is taken over the randomness of the output of the DP mechanism
and where $u_d(\cdot)$ is a monotone function for all $d\in \cD$ in the sense that $u_d(f(d)_i) \ge u_d(f(d)_{i+1})$ for $1\le i\le |\cV|-1$.

As in \cite{RafnoDP2colorPaper}, we represent a family of datasets together with their neighboring relation $(\mathcal{D},\sim)$ by a simple graph,
where the vertices are the datasets in $\mathcal{D}$ and there is an edge between $d,d' \in \mathcal{D}$ if and only if they are neighbors, i.e., $d \sim d'$.

\begin{definition}[\hspace{1sp}\cite{RafnoDP2colorPaper}]\label{def:graphmorphism}
A morphism between $(\mathcal{D}_1, \overset{1}{\sim})$ and $(\mathcal{D}_2, \overset{2}{\sim})$ is a function
\begin{align}
g: (\mathcal{D}_1, \overset{1}{\sim}) \rightarrow (\mathcal{D}_2, \overset{2}{\sim})
\end{align}
such that ${d\overset{1}{\sim}d'}$ implies in either $g(d)\overset{2}{\sim}g(d')$ or $g(d)=g(d')$ for every $d,d' \in \mathcal{D}_1$.
\end{definition}

An example of a morphism is shown in Fig.~\ref{fig: rainbow} above. A morphism $g: (\mathcal{D}_1, \overset{1}{\sim}) \rightarrow (\mathcal{D}_2, \overset{2}{\sim})$ allows to transport $(\epsilon,\delta)$-DP mechanisms from its codomain to its domain.

\begin{theorem}[\hspace{1sp}\cite{RafnoDP2colorPaper}]\label{thm:pullback}
Let $g: (\mathcal{D}_1, \overset{1}{\sim}) \rightarrow (\mathcal{D}_2, \overset{2}{\sim})$ be a morphism and $\mathcal{M}_2: \mathcal{D}_2 \rightarrow \mathcal{V}$ be an $(\epsilon,\delta)$-DP mechanism on $(\mathcal{D}_2, \overset{2}{\sim})$. Then, the mechanism $\mathcal{M}_1: \mathcal{D}_1 \rightarrow \mathcal{V}$ given by the pullback operation $\mathcal{M}_1 = \mathcal{M}_2 \circ g$ is an $(\epsilon,\delta)$-DP mechanism on $\mathcal{D}_1$.
\end{theorem}

\section{Optimal Rainbow Differential Privacy Mechanisms} \label{sec:opt}

In \cite{RafnoDP2colorPaper}, DP schemes were interpreted as randomized graph colorings. In that setting, each dataset's preference was characterized by a single color. In general, for larger output spaces, each dataset has a corresponding rainbow according to its ordering preference. Thus, we call the triple $(\mathcal{D},\sim,f)$ a \emph{rainbow graph}, where $\mathcal{D}$ is the family of datasets, $\sim$ is the neighborhood relationship, and $f: \mathcal{D} \rightarrow \Sym (\mathcal{V})$ is the preference function. We say a morphism $g\!:\!(\mathcal{D}_1, \overset{1}{\sim},f_1) \!\rightarrow\! (\mathcal{D}_2,\overset{2}{\sim},f_2)$ is \emph{rainbow-preserving} if $f_1 = f_2 \circ g$. Indeed, the morphism in Fig.~\ref{fig: rainbow} above is rainbow-preserving. We consider the following topological notions.
\begin{definition}
Let $(\mathcal{D},\sim,f)$ be a rainbow graph. Then, for every $c \in \Sym(\mathcal{V})$, we denote $B^c = \{ d\in \mathcal{D} : f(d)=c \}$. The interior of $B^c$ is the set
\begin{align}
(B^c)^\circ = \left\{ d\in B^c : d \sim d' \Rightarrow d' \in B^c \right\}
\end{align}
and its boundary is the set
\begin{align}
\partial B^c = B^c - (B^c)^\circ.
\end{align}
\end{definition}

We next study optimal DP mechanisms given a rainbow boundary condition. Since dominance (between DP mechanisms) is a partial order, there exists at most one optimal DP mechanism. In the binary case, it is known that when there exists a valid DP mechanism {on the boundary}, then there exists an optimal DP mechanism {on the whole graph}\cite{RafnoDP2colorPaper}. Surprisingly, this is no longer the case when there are more than two colors, as shown in the following example.

\begin{example} \label{eg:no-optimal}
Let $\cD = \{d_1,d_2,d_3,d_4,d_5\}$, with neighboring relations $d_1\sim d_2\sim d_3\sim d_4\sim d_5\sim d_1$.
Let $e^\epsilon=2$, $\delta=0$. Suppose the output space $\cV = \{1,2,3\}$, and the preference function $f(d_i)=(1,2,3)$ for $1\le i\le 4$ and $f(d_5)
= (1,3,2)$.
Consider the rainbow $c=(1,2,3)$. Then $B^c =\{d_1,d_2,d_3,d_4\}$ with boundary $\partial B^c = \{d_1,d_4\}$ and interior vertices $(B^c)^\circ = \{d_2,d_3\}$. Suppose the boundary condition is such that $\cM(d_1)=(0.2, 0.1, 0.7)$ and $\cM(d_4) = (0.4, 0.1, 0.5)$.
We claim that under this setting, there exist valid DP mechanisms, but there is no optimal DP mechanism.
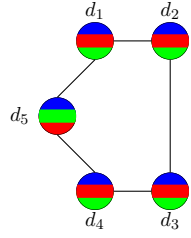
\begin{figure}[t]
\centering
\begin{tikzpicture}
    \path pic [xshift=1cm] {brg} node [align=center, xshift=1cm, above=2mm, scale=0.7, black] {$d_1$};
    \path pic [xshift=2cm] {brg} node [align=center, xshift=2cm, above=2mm, scale=0.7, black] {$d_2$};

    \path pic [xshift=0.5cm, yshift=-1cm] {bgr} node [align=center, xshift=0cm, yshift=-1.4cm, above=2mm, scale=0.7, black] {$d_5$};
   
    \path pic [xshift=1cm, yshift=-2cm] {brg} node [align=center, xshift=1cm, yshift=-2.8cm, above=2mm, scale=0.7, black] {$d_4$};
    \path pic [xshift=2cm, yshift=-2cm] {brg} node [align=center, xshift=2cm, yshift=-2.8cm, above=2mm, scale=0.7, black] {$d_3$};

    \draw (1.25,0) -- (1.75,0);
    \draw (1.0,-0.25) -- (0.5,-0.75);
    \draw (1.0,-1.75) -- (0.5,-1.25);
    \draw (1.25,-2) -- (1.75,-2);
    \draw (2,-1.75) -- (2,-0.25);

\end{tikzpicture}
\caption{\centering Illustration of the rainbow graph in Example \ref{eg:no-optimal}. We use \texttt{blue}, \texttt{red}, \texttt{green} to represent the choices $1,2,3$ respectively.}
\label{fig:my_label}
\end{figure}

First, we define two valid DP mechanisms $\cM_1$ and $\cM_2$. Let $\cM_1(d_2)=\cM_1(d_3)=(0.4, 0.2, 0.4)$, $\cM_2(d_2)=(0.4, 0.1, 0.5)$, and $\cM_2(d_3)=(0.7, 0.05, 0.25)$.
It is straightforward to verify that $\cM_1$ and $\cM_2$ are both valid DP mechanisms.
Suppose, for the sake of illustrating a contradiction, that $\cM_3$ is an optimal DP mechanism.
Because $\cM_3 \succeq \cM_1$, we have $\cM_3(d_2) \succeq \cM_1(d_2)$.
Because of the boundary condition $\cM(d_1)=(0.2, 0.1, 0.7)$, we must have $\cM_3(d_2) = \cM_1(d_2)$.
Because $\cM_3 \succeq \cM_2$, we have $\cM_3(d_3) \succeq \cM_2(d_3)$.
Because of the boundary condition $\cM(d_4)=(0.4, 0.1, 0.5)$, we must have $\cM_3(d_3) = \cM_2(d_3)$.
So we have fully determined $\cM_3$.
However, $\cM_3$ is not a valid DP mechanism because $\cM_3(d_2)_2 > e^\epsilon \cM_3(d_3)_2 + \delta$.
Thus, there exists no optimal DP mechanism.
\end{example}

The main issue in the above example is that the boundary condition is not homogeneous. That is, the conditions on datasets $d_1$ and $d_4$ are different.
We next define a homogeneity condition for DP mechanisms, generalizing \cite{RafnoDP2colorPaper} to non-binary functions.
\begin{definition}\label{def:boundaryhomogeneity}
A mechanism $\mathcal{M}: \mathcal{D} \rightarrow \mathcal{V}$ is \emph{boundary homogeneous} if, for every rainbow $c \in \Sym(\mathcal{V})$, it holds that any two boundary datasets $d,d' \in \partial B^c$ satisfy the condition $\bP[\mathcal{M}(d) = v] = \bP[\mathcal{M}(d') = v]$ for every $v \in \mathcal{V}$.
\end{definition}

We next provide our main result, which shows that under a valid homogeneous boundary condition, there exists a unique optimal DP mechanism and this optimal DP mechanism can be expressed in closed form.

\begin{theorem} \label{thm:optimal-mechanism}
Let $(\mathcal{D},\sim,f)$ be a rainbow graph and, for every rainbow $c \in \Sym(\mathcal{V})$ and $d \in \partial B^c$, let $\vec{m}^c\in \Delta (\mathcal{V})$ be a fixed homogeneous boundary condition.
Suppose that the boundary condition is valid, i.e.,
for every pair $d\sim d'\in \bigcup_{c\in \Sym(\mathcal{V})} \partial B^c$,
we have that $\vec{m}^{f(d)}$ and $\vec{m}^{f(d')}$ are $(\epsilon,\delta)$-close to each other.
Then, there exists a unique optimal $(\epsilon,\delta)$-DP mechanism satisfying the boundary condition.
\end{theorem}
Proof of Theorem~\ref{thm:optimal-mechanism} and description of the optimal DP mechanism is delayed to Section~\ref{sec:line-graph} below. We next define the notion of a line graph, which is used in the proof.

\begin{definition}
Let $c \in \Sym(\mathcal{V})$ be a rainbow and $n \in \mathbb{N}$. The $(c,n)$-line is the rainbow graph $(\mathcal{D},\sim,f)$ with datasets $\mathcal{D} = [0:n]$,
neighboring relation $i \sim j$ if $|i-j|=1$, and preference function $f(d)=c$ for every $d \in \mathcal{D}$. Dataset $0\in \mathcal{D}$ is considered the boundary of the line.\end{definition}

The other new notion we need for the proof of Theorem \ref{thm:optimal-mechanism} is that of the boundary rainbow graph of a rainbow graph, defined next.

\begin{definition}\label{def:boundarymorph}
Let $(\cD, \sim, f)$ be a rainbow graph. We define its boundary rainbow graph $(\cD_\partial, \overset{\partial}{\sim},f_\partial)$ as follows.
For $c\in \Sym(\cV)$, let
\begin{align}
d_c = \max_{d\in \cD: f(d)=c} \dist(d, \partial B^{c}),
\end{align}
where we define $\dist(x,S) = \min_{y\in \mathcal{S}} \dist(x,y)$.
That is, $d_c$ is the maximum distance of a dataset with preference $c$ to the boundary $\partial B^c$.
Let $\cD_\partial = \{(c,i): c\in \Sym(\cV), i\in [0:~d_c]\}$ and $f_\partial((c,i)) = c$.
Define $(c,i) \overset{\partial}{\sim} (c,i+1)$ for $i\in [0:d_c-1]$, and
$(c,0) \overset{\partial}{\sim} (c',0)$ (for $c,c'\in \Sym(V)$ and $c\ne c'$) if there exists $d\in B^c$, $d'\in B^{c'}$ such that $d\sim d'$.
In other words, for each preference $c$ there is a chain with $d_c+1$ vertices, with $(c,0)$ being the head and $(c,d_c)$ being the tail. There is an edge between two heads $(c,0)$, $(c',0)$ if and only if there are two adjacent datasets, one with preference $c$, the other with preference $c'$.
\end{definition}

Note that the boundary rainbow graph is a union of $(c,d_c)$-lines for $c\in \Sym(V)$ with a few possible additional edges between the endpoints $\{(c,0) : c\in \Sym(\cV)\}$; see Fig.~\ref{fig: rainbow} above. Therefore, the boundary rainbow graph consists of a series of line graphs, each for a different rainbow occurring in the original graph.
We define the boundary morphism $g_\partial: \cD \to \cD_\partial$ by sending $d\in \cD$ to $(f(d), \dist(d, \partial B^{f(d)}))\in \cD_\partial$.
We now show that optimal mechanisms for boundary-homogeneous rainbow graphs can be obtained by pulling them back via Theorem~\ref{thm:pullback} from their boundary rainbow graphs.

\begin{theorem} \label{thm:reduce-to-line}
Let $(\mathcal{D}, \sim, f)$ be a rainbow graph and  $\mathcal{M}_\partial: \mathcal{D}_\partial \rightarrow \mathcal{V}$ be the optimal $(\epsilon,\delta)$-DP mechanism on its boundary rainbow graph subject to fixed boundary probabilities. Then, the pullback $\mathcal{M} = \mathcal{M}_\partial \circ g_\partial$, where $g_\partial$ is the boundary morphism, is the optimal boundary homogeneous $(\epsilon,\delta)$-DP mechanism subject to the same boundary probabilities.
\end{theorem}

\begin{proof}
From Theorem \ref{thm:pullback}, it follows that the morphism $g_\partial: \mathcal{D} \rightarrow \mathcal{D}_{\partial}$ induces an $(\epsilon,\delta)$-DP mechanism on $\mathcal{D}$ defined by $\mathcal{M}_{\partial}\circ~g_\partial$. This mechanism is clearly boundary homogeneous.
Let $\mathcal{M}$ be a valid $(\epsilon,\delta)$-DP mechanism satisfying the boundary condition.
We prove that $\cM \preceq \mathcal{M}_{\partial} \circ g_\partial$.

Let $B^c \subseteq \mathcal{D}$ be the subset of datasets with the same preference function.
Let $d_0$ be the closest dataset in $\partial B^c$ to $d\in
B^c$. Let $G = \{d, d_{\dist(d,\partial B^c)-1}, \ldots, d_0 \}$ be a set of datasets which forms a shortest path from $d$ to $d_0$. Since $g_{\partial} \restrictedto{G}$ is injective, it has a left inverse, which we denote as $h: g_\partial (G) ~\rightarrow ~G$. However, $h$ is a morphism and, therefore, from Theorem \ref{thm:pullback}, $\mathcal{M}\circ h$ is an $(\epsilon,\delta)$-DP mechanism on $\mathcal{D}_\partial$. Then, since $\mathcal{M}_\partial$ is the optimal mechanism on $\mathcal{D}_\partial$, it follows that $\mathcal{M}_\partial \restrictedto{g_\partial (G)}$ is the optimal mechanism on $g_\partial (G)$. Thus, $\overrightarrow{\mathcal{M}_\partial \circ g_\partial }(d) \succeq \vec{\mathcal{M}}(d) $. %
Because the choice of $d$ is arbitrary, we obtain $\cM \preceq \mathcal{M}_{\partial} \circ g_\partial$, as desired.
\end{proof}

Since a boundary rainbow graph consists of a series of line graphs, the problem of finding optimal mechanisms can be reduced to finding them for line graphs, which is discussed in the next section after comparing the two orderings considered.

\section{Optimal Differentially Private Mechanisms for Line Graphs} \label{sec:line-graph}
In this section, we derive optimal DP mechanisms for line graphs and use this to prove Theorem~\ref{thm:optimal-mechanism}, and describe optimal DP mechanisms for general graphs. In this section, we use the shorthand $q := |\cV|$.

\begin{theorem}\label{thm:line-graph}
For any $(c,n)$-line graph $(\cD,\sim,f)$ with boundary condition $\vec m$, i.e.,
\begin{align}
\bP[\cM(0)=k]=m_k\quad \text{for all}\quad 1\le k\le q,
\end{align}
there exists a unique optimal DP mechanism. Furthermore, under the unique optimal DP mechanism,
$\cM(d)$, $1 \leq d \leq n$, has distribution $T_{\epsilon,\delta}^{d} (\vec m)$, where $T_{\epsilon,\delta}: \Delta(\cV)\to \Delta(\cV)$ maps $p \in \Delta(\cV)$ to $p'\in \Delta(\cV)$ defined as follows:
\begin{align}
p'_k = s'_k - s'_{k-1},
\end{align}
where we have
\begin{align}
s'_k = \min\{1, \min\{e^\epsilon s_k, 1-e^{-\epsilon} (1-s_k)\}+\delta\}\quad \text{for}\quad 1\le k\le q.
\end{align}
In the above, $s_0=s'_0 = 0$ and $s_k = \sum_{1\le i\le k} p_i$ for $1\le k\le q$.
\end{theorem}

We next describe the main results we require to prove Theorem~\ref{thm:line-graph}, whose proof is given at the end of this section. Theorem~\ref{thm:line-graph} is used to prove Theorem~\ref{thm:optimal-mechanism}, whose proof is also given at the end of this section. Now, without loss of generality, we assume that $c=(1,\ldots, |\cV|)$, and $\cV = \{1,\ldots,q\}$. Our key lemma is the following.

\begin{lemma} \label{lemma:key}
Given a distribution $p\in \Delta(\cV)$, {$p' = T_{\epsilon,\delta}(p)$ is the} unique distribution such that
\begin{enumerate}
\item $p'$ is $(\epsilon,\delta)$-close to $p$;
\item for any $p''$ that is $(\epsilon, \delta)$-close to $p$, we have $p'\succeq p''$.
\end{enumerate}
\end{lemma}

\begin{proof}
We claim that $p'$ defined above satisfies both conditions of Lemma~\ref{lemma:key}.

\textbf{Step 0.}
We verify that $p'$ is a valid distribution.
Because $s_q=1$, we have $s'_q=1$.
Because $s_k$ is monotone increasing in $k$ and both functions $e^\epsilon x$ and $1-e^{-\epsilon}(1-x)$ are monotone increasing in $x$, $s'_k$ is monotone increasing in $k$.
So $p'$ is a valid distribution.

\textbf{Step 1.}
We verify that $p'$ is $(\epsilon, \delta)$-close to $p$.
We define another distribution $\wt p$.
Let
\begin{align}
\wt s_k = \min\{e^\epsilon s_k, 1-e^{-\epsilon} (1-s_k)\}\quad  \text{for}\quad 1\le k\le q
\end{align}
and
\begin{align}
\wt p_k = \wt s_k-\wt s_{k-1}.
\end{align}
Note that $\wt p$ is a valid distribution and $\TV(\wt p,p')\le \delta$.
It suffices to prove that $\wt p$ is $(\epsilon,0)$-close to $p$.

We first prove that for any $1\le k\le q$, we have
\begin{align}
e^{-\epsilon} p_k \le \wt p_k \le e^\epsilon p_k.
\end{align}
Note that $(e^\epsilon s_k - (1-e^{-\epsilon} (1-s_k)))$ is monotone increasing in $k$.
So there exists $k_0\in \{0,\ldots,q\}$ such that
\begin{align}
\wt s_k = \left\{\begin{array}[]{ll}
e^\epsilon s_k & \text{if } k \le k_0,\\
1-e^{-\epsilon}(1-s_k) & \text{if } k \ge k_0+1.
\end{array}\right.
\end{align}
Then,
\begin{enumerate}
\item for all $k\le k_0$, we have $$\wt p_k = e^\epsilon(s_k-s_{k-1})=e^\epsilon p_k;$$
\item for all $k\ge k_0+2$, we have $$\wt p_k = (1-e^{-\epsilon}(1-s_k))-(1-e^{-\epsilon}(1-s_{k-1}))=e^{-\epsilon} p_k;$$
\item for $k=k_0+1$, we have
\begin{align*}
\wt p_k &= (1-e^{-\epsilon}(1-s_k)) - e^\epsilon s_{k-1} \le e^\epsilon s_k - e^\epsilon s_{k-1} = e^\epsilon p_k
\end{align*}
and
\begin{align*}
\wt p_k &= (1-e^{-\epsilon}(1-s_k)) - e^\epsilon s_{k-1}\\
&\ge (1-e^{-\epsilon}(1-s_k))-(1-e^{-\epsilon}(1-s_{k-1})) \!=\! e^{-\epsilon} p_k.
\end{align*}
\end{enumerate}
This proves that $\wt p$ is $(\epsilon,0)$-close to $p$.

\textbf{Step 2.}
We next prove that for any $p''$ that is $(\epsilon, \delta)$-close to $p$, we have $p'\succeq p''$.
Because $p''$ is $(\epsilon, \delta)$-close to $p$,
there exists a distribution $p^\sharp$ such that $\TV(p'',p^\sharp)\le \delta$ and that $p^\sharp$ is $(\epsilon,0)$-close to $p$. Let $s^\sharp$ be the prefix sum of $p^\sharp$ and $s''$ be the prefix sum of $p''$.

For any $1\le k\le q$, taking $\cS = \{1,\ldots,k\}$ %
in the DP condition (Definition~\ref{def:DPdef}) on $(p,p^\sharp)$, we obtain
\begin{align}
s^\sharp_k \le e^\epsilon s_k.
\end{align}
Taking $\cS = \{k+1,\ldots,q\}$, we have %
\begin{align}
s^\sharp_k \le 1-e^{-\epsilon} (1-s_k).
\end{align}
Thus, we have $s^\sharp_k \le \wt s_k$. Furthermore, because $\TV(p'',p^\sharp)\le \delta$, we have
\begin{align}
s''_k \le \min\{1,s^\sharp_k+\delta\} \le \min\{1,\wt s_k+\delta\} = s'_k.
\end{align}
Because $k$ can be any integer between $1$ and $q$, this proves that $p'\succeq p''$.
\end{proof}

\begin{lemma}\label{lemma:preserve-st}
$T_{\epsilon,\delta}$ preserves dominance ordering, i.e., $p\succeq p' \implies T_{\epsilon,\delta}(p)\succeq T_{\epsilon,\delta}(p')$.
\end{lemma}
\begin{proof}
This holds because $\min\{1, \min\{e^\epsilon x, 1-e^{-\epsilon} (1-x)\}+\delta\}$ is a monotone increasing function in $x$.
\end{proof}

Now, we are ready to prove Theorem~\ref{thm:line-graph}.
\begin{proof}[Proof of Theorem \ref{thm:line-graph}]
Let $\cM$ be the DP mechanism defined as in the theorem statement.
By Lemma \ref{lemma:key} part (1), $\cM$ is an $(\epsilon,\delta)$-DP mechanism.
Now let us prove its optimality.
Let $\cM'$ be another $(\epsilon,\delta)$-DP mechanism with the same boundary condition.

Let us use induction on $d$ to prove that $\cM(d) \succeq \cM'(d)$ for all $d\in [0:n]$.
Because of the boundary condition, we have $\cM(0)=\cM'(0)$. In particular, $\cM(0) \succeq \cM'(0)$. This is the base case of our induction.
Now suppose that we have proved $\cM(d)\succeq \cM'(d)$ for some $d\in [0:n-1]$.
Then
\begin{align}
\cM(d+1)=T_{\epsilon,\delta}(\cM(d))\succeq T_{\epsilon,\delta}(\cM'(d)) \succeq \cM'(d+1),
\end{align}
where the first step is by definition of $\cM$, the second step is by Lemma~\ref{lemma:preserve-st} and induction hypothesis, the third step is by Lemma~\ref{lemma:key}.
This completes the induction.
Therefore $\cM(d)\succeq \cM'(d)$ for all $d\in [1:n]$.
\end{proof}

Using Theorem~\ref{thm:line-graph}, we next prove Theorem~\ref{thm:optimal-mechanism}.
\begin{proof}[Proof of Theorem~\ref{thm:optimal-mechanism}]
The proof is straightforwardly obtained by combining the results of Theorems~\ref{thm:reduce-to-line} and \ref{thm:line-graph}.
Theorems~\ref{thm:reduce-to-line} shows that to construct an optimal DP mechanism for a general graph under a homogeneous boundary condition, it suffices to construct an optimal DP mechanism for the boundary rainbow graph, which is achieved by Theorem~\ref{thm:line-graph}.
\end{proof}

By expanding the proof of Theorem~\ref{thm:optimal-mechanism}, we can describe the optimal DP mechanism for general graphs as follows.
For any $c\in \Sym(\cV)$, $d\in B^c$, let $\vec m^c$ be the boundary condition. Define $\wt m^c\in \Delta(\cV)$ as $\wt m^c_{i}=\vec m^c_{c(i)}$ for $1\le i\le q$.
This is a permuted version of $\vec m^c$ such that the preference order is $(1,\ldots,|\cV|)$. Then the optimal DP mechanism is given by $\bP[\cM(d) = c(k)]=\left(T_{\epsilon,\delta}^{\dist(d, \partial B^c)} (\wt m^c)\right)_{k}$, for $1\le k\le q$.
\section{Designing the Optimal Differentially Private Mechanism} \label{sec:comp-optimal-dp}
In this section, we discuss how to design optimal DP mechanisms. Throughout, we consider a $(c,n)$-line with $c=(1,\ldots,|\cV|)$ and boundary condition $\vec m\in \Delta(\cV)$, i.e., $\bP[\cM(0)=k] = m_k$. The case $c\ne (1,\ldots,|\cV|)$ can be handled by performing a permutation, as in the end of Section~\ref{sec:line-graph}. For $t\in [n]$, the distribution of $\cM(t)$ can be straightforwardly computed using the definition of $T_{\epsilon,\delta}$ (see Theorem~\ref{thm:line-graph} above). In fact, it can be expressed in an even more explicit form, as we explain below. %

\subsection{Special Case \texorpdfstring{$\delta=0$}{delta=0}} \label{sec:comp-optimal-dp:zero}

As a warmup, let us first consider the case $\delta=0$.
By Theorem~\ref{thm:line-graph}, for the optimal DP mechanism, the distribution of $\cM(t)$ is equal to $T_{\epsilon,\delta}^{t} (\vec m)$. For $t\in \bZ_{\ge 0}$, let $s^t$ denote the prefix sum of $T_{\epsilon,\delta}^t (\vec m$).
By construction of the operator $T_{\epsilon,\delta}$, we have
\begin{enumerate}
\item if $s_k^{t} \le 1/(e^\epsilon+1)$, then $$s^{t+1}_k = e^\epsilon s^{t}_k;$$
\item if $s_k^{t} \ge 1/(e^\epsilon+1)$, then $$s^{t+1}_k = 1-e^{-\epsilon}(1-s^{t}_k).$$
\end{enumerate}
For $1\le k\le q$, define
\begin{align} \label{eqn:opt-mec-tau-delta-0}
\tau_k =\left\lfloor \max\left\{-\frac{\log (s^0_k \cdot (e^\epsilon+1))}{\epsilon}+1,0\right\}\right\rfloor.
\end{align}
Then we have $s_k^t \le 1/(e^\epsilon+1)$ for $t\le \tau_k-1$ and $s_k^t \ge 1/(e^\epsilon+1)$ for $t \ge \tau_k$. Thus, we obtain as the solution
\begin{align} \label{eqn:opt-mec-closed-form-delta-0}
s^t_k = \left\{\begin{array}[]{ll}
e^{t \epsilon} s^0_k & \text{if } t\le \tau_k,\\
1-e^{-\epsilon(t-\tau_k)}(1-s_k^{\tau_k}) & \text{if } t \ge \tau_k+1.
\end{array}\right.
\end{align}
Then, we obtain $\bP[\cM(t) = k] = s^{t}_k-s^{t}_{k-1}$.

Note that in the case of three colors ($q=3$), our expression recovers the one given in \cite{zhou2022rainbow}.

\subsection{General Case \texorpdfstring{$\delta>0$}{delta>0}} \label{sec:comp-optimal-dp:pos}
Now, we consider the case $\delta>0$. The formula is slightly more complicated, but the derivation method is similar.
Recall that $\vec m$ is the boundary mechanism.
Let $s^t$ denote the prefix sum of $T_{\epsilon,\delta}^t (\vec m$).
Define
\begin{align}
\rho = \frac{\delta}{e^\epsilon-1}.
\end{align}
By construction of the operator $T_{\epsilon,\delta}$, we have
\begin{enumerate}
\item if $s_k^t \le 1/(e^\epsilon+1)$, then $$s^{t+1}_k = \min\{1, e^\epsilon s^{t}_k + \delta\}$$
which is equal to
$$s^{t+1}_k + \rho = \min\{e^\epsilon (s^t_k + \rho), 1+\rho\};$$
\item if $s_k^t \ge 1/(e^\epsilon+1)$, then $$s^{t+1}_k = \min\{1, 1-e^{-\epsilon}(1-s^{t}_k) + \delta\}$$
which is equal to $$1-s^{t+1}_k + e^\epsilon \rho = \max\{e^{-\epsilon}(1-s^t_k + e^\epsilon \rho), e^\epsilon\rho\}.$$
\end{enumerate}

Define
\begin{align} \label{eqn:opt-mec-tau-delta-gt-0}
\tau_k =\left\lfloor \max\left\{\frac{\log \left(\left(\frac 1{e^\epsilon+1} + \rho\right)/\left(s^0_k + \rho\right)\right)}{\epsilon}+1,0\right\}\right\rfloor.
\end{align}
Then we have $s_k^t \le 1/(e^\epsilon+1)$ for $t\le \tau_k$ and $s_k^t \ge 1/(e^\epsilon+1)$ for $t\ge \tau_k+1$. Solving this, we obtain
\begin{align} \label{eqn:opt-mec-closed-form-delta-gt-0}
s^t_k = \left\{\begin{array}[]{ll}
\min\{1, e^{t\epsilon}(s^0_k + \rho) - \rho\} & \text{if } t\le \tau_k,\\
\min\{1,  1+e^\epsilon\rho - e^{-\epsilon(t-\tau_k)} (1-s_k^{\tau_k} + e^\epsilon\rho)\} & \text{if } t \ge \tau_k+1.
\end{array}\right.
\end{align}
Finally, we have $\bP[\cM(t)=k] = s_k^{t} - s_{k-1}^{t}$.

When $\delta=0$, our expressions reduce to the one we derived in Section~\ref{sec:comp-optimal-dp:zero}.

\subsection{Numerical Results}
We depict several examples for optimal DP mechanism designs in Fig.~\ref{fig:ex1},~\ref{fig:ex2}, and~\ref{fig:ex3}. In these examples, we consider $|\cV|=5$, $c=(1,\ldots,5)$, and $\epsilon=\log 1.2$.
We choose the boundary condition to be $$\vec m = (0.0005,0.0081,0.1364, 0.2727, 0.5822)$$ which corresponds to $0.0005 = 0.001\times\frac{e^\epsilon}{1+e^\epsilon}$, $0.0081 = 0.015 \times\frac{e^\epsilon}{1+e^\epsilon}$, $0.1364 = 0.25\times\frac{e^\epsilon}{1+e^\epsilon}$, $0.2727 = 0.5\times\frac{e^\epsilon}{1+e^\epsilon}$, and $0.5822$ is 1 minus all the other values. %

Fig.~\ref{fig:ex1} illustrates the optimal $(\epsilon,\delta)$-DP mechanism for $\delta=0$.
At integer $t$, the figure shows $T_{\epsilon,\delta}^t \vec m$.
At non-integer $t$, the values are interpolated using \eqref{eqn:opt-mec-closed-form-delta-0}.
By \eqref{eqn:opt-mec-closed-form-delta-0}, we observe that for fixed $k$, $s_k^t$ goes through a phase transition at $\tau_k$. For $t\le \tau_k$, $s_k^t$ increases exponentially and for $t\ge \tau_k$, $1-s_k^t$ decreases exponentially, respectively. By \eqref{eqn:opt-mec-tau-delta-0}, $\tau_k$ is non-increasing in $k$. Therefore, the probabilities $\bP[\cM(t)=k] = s_k^{t} - s_{k-1}^{t}$ goes through at most three phases as $t$ increases. We have the following results:
\begin{itemize}
\item in the first phase with $t \le \tau_k$, $s_k^t$ and $s_{k-1}^t$ both increase exponentially;
\item in the second phase with $\tau_k\le t\le \tau_{k-1}$, $s_{k-1}^t$ increases exponentially, while $1-s_k^t$ decreases exponentially;
\item in the third phase with $t\ge \tau_{k-1}$, $1-s_k^t$ and $1-s_{k-1}^t$ both decrease exponentially.
\end{itemize}
When $\tau_k=0$, the first phase is degenerate (i.e., has length $0$);
when $\tau_k=\tau_{k-1}$, the second phase is degenerate; and
when $\tau_{k-1}=\infty$ (i.e., $s_{k-1}^0=0$), the third phase is degenerate.
In Fig.~\ref{fig:ex1}, the first phase for $k=5$ is degenerate (no phase transition), and the third phase for $k=1$ is degenerate.

\begin{figure}[t]
\centering

\includegraphics[scale=0.6]{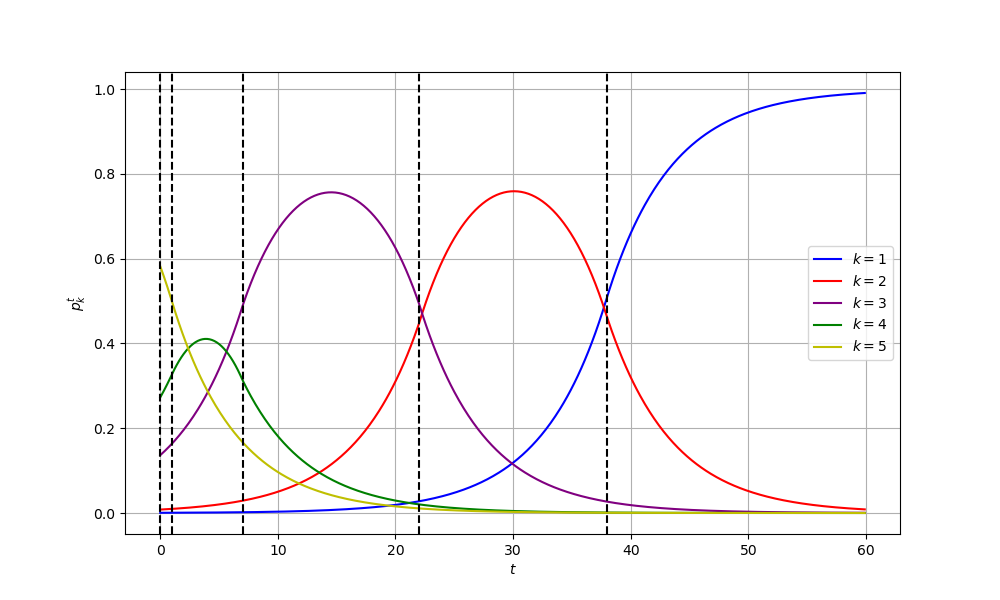}
\caption{The optimal $(\log (1.2),0)$-DP mechanism with homogeneous boundary condition $\vec m =(0.0005,0.0081,0.1364, 0.2727, 0.5822)$. We have $\tau_1= 38, \tau_2= 22, \tau_3= 7, \tau_4= 1,$ and $\tau_5=0$.}%
\label{fig:ex1}
\end{figure}

Secondly, Fig.~\ref{fig:ex2} shows the optimal DP mechanisms with the same parameters as Fig.~\ref{fig:ex1}, but with $\delta = 10^{-3}$.
At integer $t$, the figure shows $T_{\epsilon,\delta}^t \vec m$. At non-integer $t$, the values are interpolated using \eqref{eqn:opt-mec-closed-form-delta-gt-0}.
From \eqref{eqn:opt-mec-closed-form-delta-gt-0} we see that for fixed $k$, $s_k^t$ goes through a phase transition at $\tau_k$.
For $t\le \tau_k$, $s_k^t$ increases exponentially (with a drift); and for $t\ge \tau_k$, $1-s_k^t$ decreases exponentially (with a drift). The phase transition behavior is similar to the $\delta=0$ case. %

\begin{figure}[t]
\centering
\includegraphics[scale=0.6]{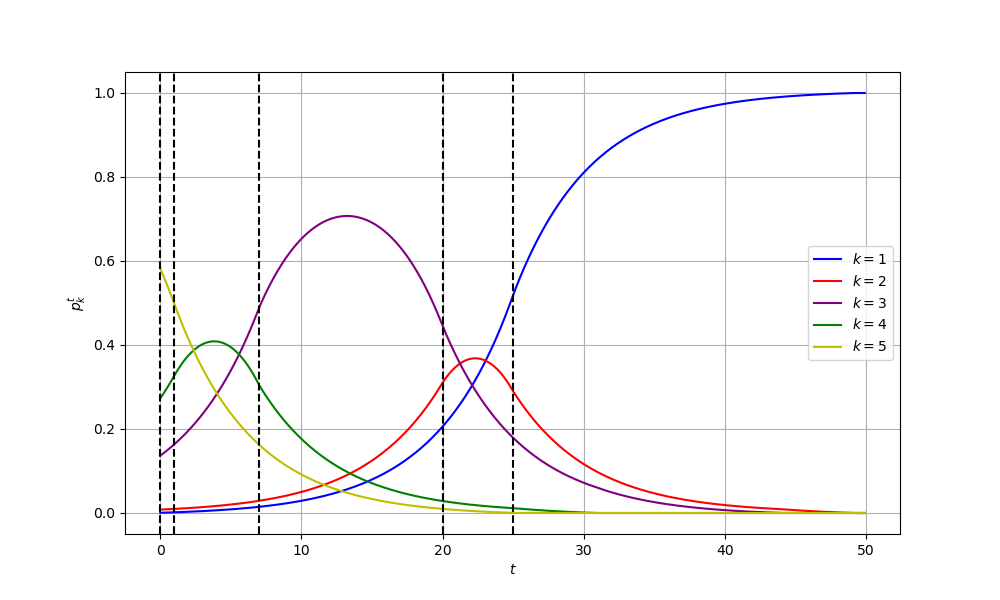}
\caption{The optimal $(\log(1.2),10^{-3})$-DP mechanism with homogeneous boundary condition $\vec m =(0.0005,0.0081,0.1364, 0.2727, 0.5822)$. We have $\tau_1= 25, \tau_2= 20, \tau_3= 7, \tau_4= 1,$ and $\tau_5=0$.}%
\label{fig:ex2}
\end{figure}

Finally, Fig.~\ref{fig:ex3} shows the optimal DP mechanisms with the same parameters as Fig.~\ref{fig:ex1} and Fig.~\ref{fig:ex2}, but with $\delta=0.01$. In this case, the phase transitions happen earlier (i.e., the $\tau$'s are smaller than previous examples).

\begin{figure}[t]
\centering
\includegraphics[scale=0.6]{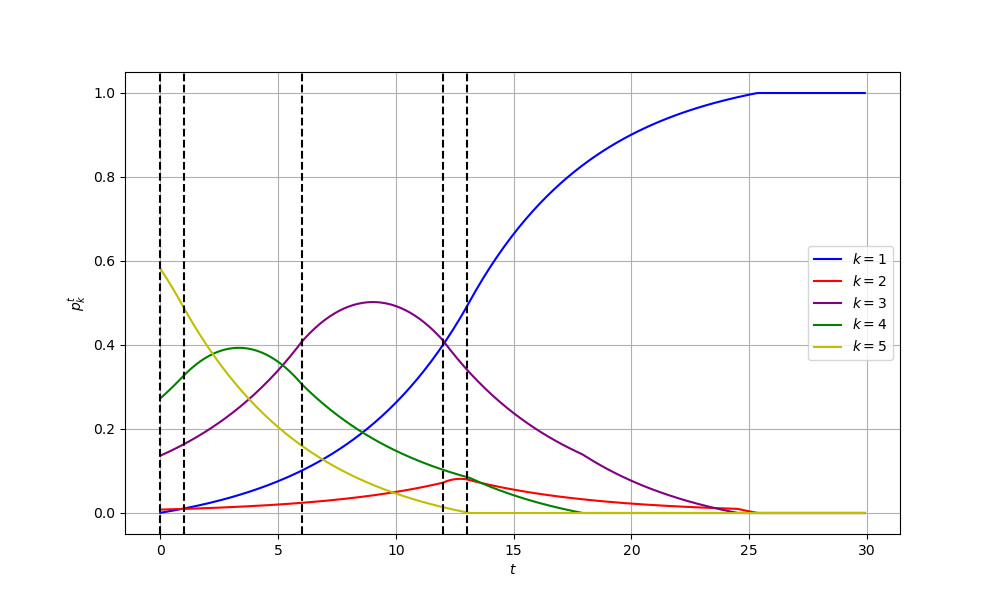}
\caption{The optimal $(\log (1.2),0.01)$-DP mechanism with the homogeneous boundary condition $\vec m =(0.0005,0.0081,0.1364, 0.2727, 0.5822)$. We have
$\tau_1= 13, \tau_2= 12, \tau_3= 6, \tau_4= 1,$ and $\tau_5=0.$}%
\label{fig:ex3}
\end{figure}

\section{Discussions} \label{sec:discussion}
\subsection{Contributions}
In this paper, we presented optimal rainbow DP mechanisms given valid homogeneous boundary conditions for any finite query output sizes by using a new proof technique.

We remark that it is a priori not clear whether an optimal mechanism exists for either homogeneous or non-homogeneous boundary conditions. In fact, our Example~\ref{eg:no-optimal} shows that there exist examples with an inhomogeneous boundary condition where there are valid mechanisms but there is no optimal mechanism. Therefore, we believe it is an interesting and non-trivial result to show that under homogeneous boundary condition, an optimal mechanism exists, and it is optimal under a wide range of utility functions. Furthermore, we give an explicit formula for the optimal mechanism. As shown in previous works \cite{RafnoDP2colorPaper,zhou2022rainbow}, as the number of colors grows, the complexity of the optimal DP mechanism also increases, and it is a priori not clear whether the complexity will be out of reach as the number of colors become larger. Our result shows that the answer is no, as we give a uniform treatment for any number of colors.

\subsection{Dataset dependency}
Our mechanism is a dataset-dependent mechanism.
In general, dataset dependency is desirable from the perspective of DP mechanism design because it allows more room for mechanism design and may lead to improved utility. In our rainbow DP setting, a dataset-independent mechanism would be suboptimal because such a mechanism will not be aware of whether a dataset is close to or far away from the boundary. A dataset-dependent mechanism such as our optimal mechanism can be more aggressive when choosing the output distribution for datasets far away from the boundary. Our result advocates using data-dependent mechanisms in DP mechanism design.

\subsection{Lexicographic vs. Dominance Ordering}
In \cite{zhou2022rainbow}, optimal rainbow DP mechanisms under lexicographic ordering are studied, while we focus on dominance ordering. We briefly discuss the relationship between the two orderings.

On the probability simplex $\Delta(\cV)$, dominance ordering is strictly stronger than lexicographic ordering. Therefore, a rainbow DP mechanism that is optimal under dominance ordering is also optimal under lexicographic ordering. As we show in Theorem~\ref{thm:optimal-mechanism} above, under homogeneous boundary conditions, there exists a unique optimal rainbow DP mechanism under dominance ordering, which implies the same result for lexicographic ordering. However, under non-homogeneous boundary conditions, as we illustrate in Example~\ref{eg:no-optimal} above, there exist scenarios for which there is no optimal rainbow DP mechanism, under either lexicographic ordering or dominance ordering. An interesting question is whether there exist scenarios for which there is an optimal rainbow DP mechanism under lexicographic ordering, but no optimal mechanism under dominance ordering. Our efforts in this direction have not resulted in any such scenario. We know that such an example, if it exists, should have a non-homogeneous boundary condition. However, it seems difficult to control a DP mechanism to be optimal under lexicographic ordering, unlike under dominance ordering.

\subsection{Comparison with exponential mechanism}
We compare the exponential mechanism with our mechanism on the rainbow DP problem. First of all, because our mechanism is optimal, any mechanism produced using the exponential mechanism cannot be better than ours. So the crux of the question is whether our mechanism can be produced using the exponential mechanism. We note that a dataset-independent version cannot produce our mechanism, because our mechanism is dataset-dependent. If we consider dataset-dependent versions, then the question is how to design the score function used in the exponential mechanism. A clever design of the score function could give the same mechanism as ours, but it seems like to find such a clever design, it is necessary to use proofs similar to ours, rather than using the usual proof of utility for exponential mechanisms. Therefore, using exponential mechanism in this setting seems to have little benefit.

It is an interesting question whether our mechanism can be applied to other problems such as median queries. It might be related to a continuous generalization we will discuss below. We would leave this direction for future work.

\subsection{Further directions}
A possible extension to our result is the case of continuous alphabets (e.g., an interval $[0, 1]$). For example, when we have $\cV=[0,1]$ and smaller numbers are preferred over larger numbers, we could replace dominance ordering with a suitable generalization (e.g., first-order stochastic dominance \cite{hadar1969rules}), for which case generalizations of our main results (Theorems~\ref{thm:optimal-mechanism} and \ref{thm:line-graph}) are expected to hold. We leave this direction for further research.

\section*{Acknowledgment}
This work has been supported in part by NSF DMS-1926686, the German Federal Ministry of Education and Research (BMBF) within the National Initiative on 6G Communication Systems through the Research Hub \emph{6G-life} under Grant 16KISK001K, the German Research Foundation (DFG) as part of Germany’s Excellence Strategy – EXC 2050/1 – Project ID 390696704 – Cluster of Excellence \emph{``Centre for Tactile Internet with Human-in-the-Loop'' (CeTI)} of Technische Universit\"at Dresden, the ARC Future Fellowship FT190100429, by the ELLIIT funding endowed by the Swedish government, by the ZENITH Research and Leadership Career Development Award fund, by NSF RINGS-2148132, and NSF CNS-2008624.
We thank anonymous reviewers for helpful suggestions and comments.
We thank Ying Zhao for pointing out a numerical issue in a figure in a previous version of the paper.

\bibliographystyle{ieeetr}
\bibliography{ref}

\vfill

\end{document}